\documentclass[11pt]{amsart}
\usepackage{amsfonts}
\usepackage{amsmath}
\usepackage{amssymb}
\usepackage{color}
\usepackage[all]{xy}
\usepackage{graphicx}
\usepackage{xspace}
\usepackage{wrapfig}
\usepackage{axodraw4j}

 \font \eightrm=cmr8

 \newcommand{\nc}{\newcommand}

\newtheorem{thm}{Theorem}
\newtheorem{exam}{Example}

\newtheorem{cor}[thm]{Corollary}

\newtheorem{lem}[thm]{Lemma}
\newtheorem{prop}[thm]{Proposition}
\newtheorem{defn}{Definition}

\def\IIpointiloop{{\scalebox{0.15}{ 
   \begin{picture}(150,45) (120,-45)
    \SetWidth{1.0}
    \SetColor{Black}
    \Line(120,-45)(270,-45)
    \CArc(195,-45)(45,-0,180)
  \end{picture}
}}}

\def\IIpointiiloop{{\scalebox{0.15}{ 
  \begin{picture}(150,45) (120,-45)
    \SetWidth{1.0}
    \SetColor{Black}
    \Line(120,-45)(270,-45)
    \CArc(195,-45)(45,-0,180)
    \CArc(195,-45)(30,-0,180)
  \end{picture}
}}}

\def\IIpointiiiloop{{\scalebox{0.15}{ 
   \begin{picture}(150,75) (120,-45)
    \SetWidth{1.0}
    \SetColor{Black}
    \Line(120,-15)(270,-15)
    \CArc(195,-15)(45,-0,180)
    \CArc(195,37.5)(37.5,-143.13,-36.87)
    \CArc(195,-15)(30,-180,0)
  \end{picture}
}}}

\def\IIIpointiloop{{\scalebox{0.13}{ 
  \begin{picture}(150,120) (120,-30)
    \SetWidth{1.0}
    \SetColor{Black}
    \Line(150,30)(270,-30)
    \Line(150,30)(270,90)
    \Line(240,-15)(240,75)
    \Line(120,30)(150,30)
  \end{picture}
}}}

\def\IIIpointiiloopA{{\scalebox{0.13}{ 
\begin{picture}(150,120) (120,-30)
    \SetWidth{1.0}
    \SetColor{Black}
    \Line(150,30)(270,-30)
    \Line(150,30)(270,90)
    \Line(240,-15)(240,75)
    \Line(120,30)(150,30)
    \Line(210,0)(210,60)
  \end{picture}

}}}

\def\IIIpointiiloopB{{\scalebox{0.13}{ 
  \begin{picture}(150,120) (120,-30)
    \SetWidth{1.0}
    \SetColor{Black}
    \Line(150,30)(270,-30)
    \Line(150,30)(270,90)
    \Line(240,-15)(240,75)
    \Line(120,30)(150,30)
    \CArc(240,30)(30,90,270)
  \end{picture}
}}}

\def\IVpointiloop{{\scalebox{0.15}{ 
    \begin{picture}(150,60) (120,-45)
    \SetWidth{1.0}
    \SetColor{Black}
    \CArc(195,-33.75)(48.75,22.62,157.38)
    \CArc(195,3.75)(48.75,-157.38,-22.62)
    \Line(120,-30)(150,-15)
    \Line(120,0)(150,-15)
    \Line(240,-15)(270,0)
    \Line(240,-15)(270,-30)
  \end{picture}

}}}

\def\IVpointiiloop{{\scalebox{0.15}{ 
  \begin{picture}(150,119) (120,-15)
    \SetWidth{1.0}
    \SetColor{Black}
    \Line(120,59)(270,-16)
    \Line(120,29)(270,104)
    \CArc(300,44)(75,143.13,216.87)
    \CArc(180,44)(75,-36.87,36.87)
  \end{picture}}}}

\def\IgraphBox{{\scalebox{0.35}{ 
 \begin{picture}(194,110) (31,-41)
    \SetWidth{1.0}
    \SetColor{Black}
    \CBox(89,25)(168,68){Black}{White}
    \CBox(89,-40)(168,3){Black}{White}
    \SetWidth{1.5}
    \Arc[clock](128,65)(35.777,-26.565,-153.435)
    \Line(32,14)(80,14)
    \Line(176,14)(224,14)
    \Arc(128,14)(48,180,540)
    \Arc(128,-37)(35.777,26.565,153.435)
  \Text(62,30)[lb]{\Huge{\Black{$\Gamma$}}}
    \Text(120,-25)[lb]{\LARGE{\Black{$\gamma_1$}}}
    \Text(120,40)[lb]{\LARGE{\Black{$\gamma_2$}}}
  \end{picture}}}}
 
 \def\IgraphBoxA{{\scalebox{0.35}{ 
  \begin{picture}(194,104) (31,-41)
    \SetWidth{1.0}
    \SetColor{Black}
    \CBox(89,19)(168,62){Black}{White}
    \SetWidth{1.5}
    \Arc[clock](128,59)(35.777,-26.565,-153.435)
    \Line(32,8)(80,8)
    \Line(176,8)(224,8)
    \Arc(128,8)(48,180,540)
    \Arc(128,-43)(35.777,26.565,153.435)
    \Text(62,30)[lb]{\Huge{\Black{$\Gamma$}}}
    \Text(120,-32)[lb]{\LARGE{\Black{$\gamma_1$}}}
    \Text(120,34)[lb]{\LARGE{\Black{$\gamma_2$}}}
  \end{picture}}}}
 
 \def\IgraphBoxB{{\scalebox{0.35}{ 
  \begin{picture}(194,104) (31,-47)
    \SetWidth{1.0}
    \SetColor{Black}
    \CBox(89,-46)(168,-3){Black}{White}
    \SetWidth{1.5}
    \Arc[clock](128,59)(35.777,-26.565,-153.435)
    \Line(32,8)(80,8)
    \Line(176,8)(224,8)
    \Arc(128,8)(48,180,540)
    \Arc(128,-43)(35.777,26.565,153.435)
     \Text(62,30)[lb]{\Huge{\Black{$\Gamma$}}}
    \Text(120,-32)[lb]{\LARGE{\Black{$\gamma_1$}}}
    \Text(120,34)[lb]{\LARGE{\Black{$\gamma_2$}}}  
\end{picture}}}}

\def\IgraphGamma{{\scalebox{0.20}{ 
  \begin{picture}(194,110) (31,-41)
    \SetWidth{1.0}
    \SetColor{Black}
    \SetWidth{1.5}
    \Arc[clock](128,65)(35.777,-26.565,-153.435)
    \Line(32,14)(80,14)
    \Line(176,14)(224,14)
    \Arc(128,14)(48,180,540)
    \Arc(128,-37)(35.777,26.565,153.435)
    \Text(120,-25)[lb]{\Huge{\Black{$\gamma_1$}}}
    \Text(120,38)[lb]{\Huge{\Black{$\gamma_2$}}}
  \end{picture}}}}

\def\IgraphB{{\scalebox{0.35}{
 \begin{picture}(194,98) (31,-47)
    \SetWidth{1.5}
    \SetColor{Black}
    \Line(32,2)(80,2)
    \Line(176,2)(224,2)
   \Arc(128,2)(48,180,540)
   \Arc(128,-49)(35.777,26.565,153.435)
    \Text(5,18)[lb]{\Huge{\Black{$\Gamma/S_{2}$}}}
    \Text(120,-34)[lb]{\LARGE{\Black{$\gamma_1$}}}   
  \end{picture}}}}

\def\IgraphC{{\scalebox{0.35}{
  \begin{picture}(194,98) (31,-47)
    \SetWidth{1.5}
    \SetColor{Black}
    \Line(32,2)(80,2)
    \Line(176,2)(224,2)
   \Arc(128,2)(48,180,540)
    \Text(5,18)[lb]{\Huge{\Black{$\Gamma/S_{12}$}}}
  \end{picture}}}}

\nc{\ignore}[1]{{}}
\nc{\mrm}[1]{{\rm #1}}
\nc{\dirlim}{\displaystyle{\lim_{\longrightarrow}}\,}
\nc{\invlim}{\displaystyle{\lim_{\longleftarrow}}\,}
\nc{\vep}{\varepsilon} \nc{\ep}{\epsilon}
\nc{\sigmat}{\widetilde\sigma}
\nc{\ostar}{\overline{*}}

\nc{\mchar}{\mrm{Char}}
\nc{\Hom}{\mrm{Hom}}
\nc{\id}{\mrm{id}}

\nc{\remark}{\noindent{\bf{Remark:}}}
\nc{\remarks}{\noindent{\bf{Remarks:}}}

 \nc{\delete}[1]{}
 \nc{\grad}[1]{^{({#1})}}
 \nc{\fil}[1]{_{#1}}

\nc{\BA}{{\Bbb A}} \nc{\CC}{{\Bbb C}} \nc{\DD}{{\Bbb D}}
\nc{\EE}{{\Bbb E}} \nc{\FF}{{\Bbb F}} \nc{\GG}{{\Bbb G}}
\nc{\HH}{{\Bbb H}} \nc{\LL}{{\Bbb L}} \nc{\NN}{{\Bbb N}}
\nc{\PP}{{\Bbb P}} \nc{\QQ}{{\Bbb Q}} \nc{\RR}{{\Bbb R}}
\nc{\TT}{{\Bbb T}} \nc{\VV}{{\Bbb V}} \nc{\ZZ}{{\Bbb Z}}

\nc{\Cal}[1]{{\mathcal {#1}}}
\nc{\mop}[1]{\mathop{\hbox {\rm #1} }}
\nc{\smop}[1]{\mathop{\hbox {\eightrm #1} }}
\nc{\mopl}[1]{\mathop{\hbox {\rm #1} }\limits}
\nc{\frakg}{{\frak g}}
\nc{\g}[1]{{\frak {#1}}}
\def \restr#1{\mathstrut_{\textstyle |}\raise-8pt\hbox{$\scriptstyle #1$}}
\def \srestr#1{\mathstrut_{\scriptstyle |}\hbox to
  -1.5pt{}\raise-4pt\hbox{$\scriptscriptstyle #1$}}
\nc{\wt}{\widetilde}
\nc{\wh}{\widehat}
\nc{\un}{\hbox{\bf 1}}
\nc{\redtext}[1]{\textcolor{red}{\tt #1}}
\nc{\bluetext}[1]{\textcolor{blue}{#1}}
\nc{\comment}[1]{[[{\tt {#1}}]] }
\nc{\R}{{\mathbb R}}

\nc\fleche[1]{\mathop{\hbox to #1 mm{\rightarrowfill}}\limits}
\def\semi{\mathrel{\times}\kern -.85pt\joinrel\mathrel{\raise 1.4pt\hbox{${\scriptscriptstyle |}$}}}

\begin{document}

\title[Exponential renormalisation and Bogoliubov's $R$-operation]
{{\Large{Exponential Renormalisation II}} \\[0.2cm] 
{\small{Bogoliubov's $R$-operation and momentum subtraction schemes\\ \phantom{m}}}}

\author{Kurusch Ebrahimi-Fard}
\address{Instituto de Ciencias Matem\'aticas,
		C/ Nicol\'as Cabrera, no.~13-15, 28049 Madrid, Spain.
		On leave from Univ.~de Haute Alsace, Mulhouse, France}
         \email{kurusch@icmat.es, kurusch.ebrahimi-fard@uha.fr}         
         \urladdr{www.icmat.es/kurusch}

\author{Fr\'ed\'eric Patras}
\address{Laboratoire J.-A.~Dieudonn\'e
         		UMR 6621, CNRS,
         		Parc Valrose,
         		06108 Nice Cedex 02, France.}
\email{patras@math.unice.fr}
\urladdr{www-math.unice.fr/~patras}

\date{April, 2nd 2012}

\begin{abstract}
This article aims at advancing the recently introduced exponential method for renormalisation in perturbative quantum field theory. It is shown that this new procedure provides a meaningful recursive scheme in the context of the algebraic and group theoretical approach to renormalisation. In particular, we describe in detail a Hopf algebraic formulation of Bogoliubov's classical $R$-operation and counterterm recursion in the context of momentum subtraction schemes. This approach allows us to propose an algebraic classification of different subtraction schemes. Our results shed light on the peculiar  algebraic role played by the degrees of Taylor jet expansions, especially the notion of minimal subtraction and oversubtractions.        
\end{abstract}

\maketitle
\tableofcontents

\section{Introduction}
\label{sect:intro}

Let us start with A.~S.~Wightman's characterisation of mathematical physics as ``the pursuit of significant structure in physical theory"\footnote{From the editor's foreword in Mackey's monograph on the Mathematical Foundations of Quantum Mechanics \cite{Mackey}.}. D.~Kreimer recently \cite{Kreimer} uncovered a Hopf algebra structure underlying Feynman graphs. This insight marked the starting point of a refined approach to the understanding of the combinatorics of renormalisation in perturbative quantum field theory (QFT). It has also led to various new developments in mathematics. Indeed, the mathematical community realized soon after Kreimer's work that renormalisation techniques familiar in perturbative QFT, such as minimal subtraction (MS) in dimensional regularisation (DR) also make sense, for example, to study singularities of hypergeometric functions, or to define rough paths. The existence of underlying Hopf algebra structures is essential to these extensions of QFT techniques beyond their usual application domains.

In this spirit, the present article sets out to further investigate the fine (Hopf) algebraic structure of the renormalisation process in perturbative QFT. It focusses on ``physical'' renormalisation schemes, such as momentum subtraction schemes. Here the external structure, i.e.~external momenta, of Feynman graphs is taken into account in the definition of the subtraction maps -- e.g.~in order to fix renormalisation conditions. One of our principal goals is the extension of the elegant Connes--Kreimer Hopf algebra approach to perturbative renormalisation \cite{CKI,CKII,CKIII} to such schemes. The motivation for this stems from the observation that the Rota--Baxter condition imposed on the subtraction map, which is essential in the aforementioned works, is not satisfied in general. In fact, it applies straightforwardly only in the context of the MS scheme in DR. 

The exponential renormalisation method introduced in \cite{EFP} allows to include multiplicative renormalisation into the Hopf algebraic framework of renormalisation, since the underlying recursive process can be understood as a recursive reparametrisation of the Lagrangian of the theory. Relying on a slightly modified version of the method, we show here that, even in the case when a proper Rota--Baxter structure is not available, one can still obtain recursively a meaningful Birkhoff--Wiener--Hopf (BWH) type decomposition of Hopf algebra characters. Hence, a group-theoretical construction of the counterterm and renormalised Feynman rules. We show that when the underlying recursion is properly defined, so as to take into account the features of the momentum subtraction maps including the subtraction degree, this decomposition agrees with Bogoliubov's formulae. These results are encoded in Theorem~\ref{thm:mainExpRenII}. 
 
Concretely, we describe in detail a Hopf algebraic formulation of Bogoliubov's classical $R$-operation and counterterm recursion \cite{BoPa} in the context of Taylor jet expansions. To motivate our procedure we show briefly how the Hopf algebraic approach to renormalisation implies an algebraic constraint on subtraction maps. Later, we introduce several simple but seemingly new ideas in order to improve the mathematical account of the theory of renormalisation.

In particular we would like to underline that the here presented exponential method allows us to propose a purely algebraic classification of renormalisation schemes. The strong algebraic flavor of this classification is mainly motivated by the Hopf algebraic framework used to reorganize renormalisation. Regarding momentum subtraction schemes, one of the interesting aspects is the distinction between the notions of minimal and oversubtraction from an algebraic point of view. This important difference is reflected in our classification, and sheds light on the joint use of the two schemes in the classical QFT literature. Recall in particular the introduction of the notion of oversubtraction in the context of the proof of the forest formula by Zimmermann \cite{Zimmermann1,Zimmermann2}, which has been one of the cornerstones in the modern treatment of renormalisation -- see e.g.~Chetyrkin's historical account in the introduction of \cite{Chetyrkin}.

At the end we remark that the present algebraic approach leads naturally to a new construction of the counterterm as well as of the corresponding BWH decomposition, with properties one expects from such a decomposition. Though it differs slightly from the one implied by Bogoliubov's $R$-operation in a straightforward sense. In this article we solely focus on foundational aspects of the exponential renormalisation method, and intend to address more specific issues in the context of these new ideas on algebraic structures of renormalisation schemes in forthcoming works.

The paper is organised as follows. In section \ref{sect:RoperationHopf} we recall the basics of Feynman graph calculus as well as Taylor jet subtractions. Bogoliubov's $R$-operation together with momentum subtraction schemes are introduced briefly. In the light of the Hopf algebraic approach to perturbative renormalisation, section \ref{sect:schemes} contains a tentative proposal for an algebraic classification of subtraction schemes. Section \ref{sect:ExpRen} elaborates on a slightly modified version of the recently introduced exponential method of perturbative renormalisation in the context of the aforementioned classification.


\section{Bogoliubov's $R$-operation and Feynman graphs}
\label{sect:RoperationHopf}

In the following we fix the notations and briefly recall how Feynman amplitudes are computed. Then we introduce Bogoliubov's $R$-operation and counterterm recursion. The setting is developed for momentum space renormalization and euclidean signature.\\


{\bf{Preliminaries}}. Perturbation theory is effectively expressed using Feynman graphs. From a Lagrangian function one can derive Feynman rules. Via these rules any Feynman graph~$\Gamma$ corresponds to a Feynman amplitude $J_\Gamma$. For details see e.g.~\cite{IZ} or any other QFT textbook. In the following we denote this correspondence by $\phi$. It is extended to a linear form on the polynomial algebra $H$ over the set $F$ of ---one particle irreducible (1PI) ultraviolet (UV) divergent--- Feynman graphs by:
$$
	\phi(\Gamma_1 \cdots \Gamma_k):=\phi(\Gamma_1) \cdots \phi(\Gamma_k)
	=J_{\Gamma_1} \cdots J_{\Gamma_k},
$$ 
where $\Gamma_1, \ldots, \Gamma_k$ is are 1PI Feynman graphs. Observe that the product of these graphs on the left hand side is the one in the (Hopf) algebra $H$, whereas the product on the right hand side is the product of amplitudes in the underlying field $k=\mathbb{C}$.

In general $J_\Gamma$ is a multiple $D(=4)$-dimensional momentum space integral multiplied by a certain power of the coupling constant:
\begin{equation}
	\Gamma \xrightarrow{\phantom{mm}\phi \phantom{mm}}\phi(\Gamma)(p^\Gamma,s;g)
			:= J_\Gamma(p^\Gamma,s;g) 
			 = g^{V_\Gamma} \int\prod_{i=1}^{L_\Gamma}\,d^Dk_i \ I_\Gamma(p^\Gamma,k^\Gamma,s).
\label{eq:stone-of-contention}
\end{equation}
Here $V_\Gamma$ is the number of vertices of the graph $\Gamma$ (for notational simplicity, we consider the case where there is a unique coupling constant). $L_\Gamma=:|\Gamma|$ denotes the number of loops in the diagram, and $k^\Gamma:=(k_1,\ldots,k_{|\Gamma|})$ are the corresponding independent internal (loop) momenta, i.e.~each independent loop yields one integration. The external momenta are denoted by $p^\Gamma:=(p_1,\ldots,p_n)$, with overall momentum conservation, $\sum_{i=1}^n p_i=0$. They are represented by external legs of the graph $\Gamma$. Other parameters, such as masses, are collected in $s$. Note that in any given theory, for each Feynman graph there exists a rigid relation between its numbers of loops and vertices. For instance, in the $\varphi^4_4$-model, for graphs with two external legs the number of vertices just equals the number of loops. However, for graphs with four external legs the number of vertices equals the number of loops plus one, and so on. We refer the reader to the standard references for details. 

The integrand $I_\Gamma(p^\Gamma,k^\Gamma,s)$ is a product of propagators and vertex terms. For most (i.e.~$D=4$, renormalisable) quantum field theories (with dimensionless couplings), these integrals suffer from UV divergencies. Note that in this work we ignore the problem of infrared (IR) divergencies completely. Concretely, under a scale transformation, the integrand behaves as:
$$
	\bigg[\prod_{l=1}^{|\Gamma|}d^D(\lambda k_l)\bigg] I_\Gamma(\lambda p^\Gamma,\lambda k^\Gamma,s) \sim \lambda^{\omega(\Gamma)},
$$
with~$\omega(\Gamma)$ the overall UV degree of divergence of the graph~$\Gamma$. Power-counting renormalisable theories are such that all interaction terms in the Lagrangian are of dimension smaller than or equal to $D$; then $\omega(\Gamma)$ is bounded by a number independent of the order of the graph, i.e.~its number of vertices.  For the sake of simplicity we only consider scalar bosonic theories to illustrate this. For instance in the $\Phi^4_4$-model the overall UV degree of divergence of a graph with~$N$ external legs is:
$$
    \omega(\Gamma) = 4L - 2l = 2l - 4V + 4 = 4 - N.
$$
Here $L$, $l$ and $V$ are the numbers of loops, internal lines and vertices of the graph $\Gamma$, respectively. The second equality follows from the relation between the numbers of loops and vertices in a graph, i.e. $L = l - V + 1$. The last equality, saying that $\omega(\Gamma)$ only depends on the number of external legs, reflects the fact that the $\Phi^4_4$-model is renormalisable in four dimensions. Another example is the $\Phi^6_3$-model, where the overall UV degree of divergence of a graph with~$N$ external legs is $ \omega(\Gamma) =  6 - 2N$. The celebrated Weinberg--Zimmermann theorem states that, provided all free propagators have nonzero masses, the integral associated to the Feynman graph~$\Gamma$ is absolutely convergent if its overall UV degree of divergence and that of each of its 1PI subgraphs are strictly negative.

The classical BPHZ momentum space subtraction method, which will be largely the subject of this article from an algebraic point of view, is rooted in this assertion \cite{BoPa,CasKen,Chetyrkin,IZ,Vasilev2,Zavialov2,Zimmermann1,Zimmermann2}. At the one loop level, the essential idea is to redefine the integrand $I_\Gamma(p^\Gamma,k^\Gamma,s)$ of a divergent integral by subtracting the first $\omega(\Gamma)$ terms of its Taylor expansion in the external momenta~$p^\Gamma$ at an appropriate subtraction point. Beyond the one loop level, Zimmermann proved that such subtractions, performed on renormalisation parts, i.e.~integrands corresponding to 1PI UV divergent subgraphs of~$\Gamma$, lower the UV degrees of the integral and its subintegrals until they become all negative, and hence convergent. The proper combinatorics of these subgraph subtractions is encoded by Bogoliubov's $R$- respectively $\bar R$-operation, a recursion solved by Zimmermann's forest formula.\\


{\bf{Bogoliubov's $R$-operation}}. Let us briefly recall Bogoliubov's $R$-operation. With the goal to simplify the presentation, we refrain from addressing questions related to the proper choice of subtraction points in Taylor jets as well as IR divergencies, and assume to work with a theory which allows for subtractions around zero momentum.

Recall that, for mathematical consistency, one must first render the UV divergent Feynman integrals formally finite upon the introduction of new non-physical parameters \cite{CasKen}. For instance, one might truncate, i.e., cut-off the integration limits at an upper bound $\Lambda$. Evaluating such a regularised integral results in terms containing $\Lambda$ in such a way that naively removing the cut-off parameter must be avoided since it gives back the original divergences. In general, for prescribed values of these regularisation parameters, we retrieve the original ill-defined amplitude. In the following we implicitly assume some proper regularisation to be at work.

Bogoliubov's $R$-operation consists of an elaborate subtraction procedure to be applied to regularised Feynman amplitudes, such that upon removing the regularisation parameter, it gives physically sound and finite expressions. Since its inception a precise understanding of the functioning of the $R$-operation has been mandatory. See \cite{CasKen,Collins1,Collins2,IZ,Vasilev2,Zavialov1} for more details on the $R$-operation, including brief and instructive accounts. It was explored in detail especially by the Russian school. The textbooks by V.~A.~Smirnov \cite{Smirnov}, A.~N.~Vasil'ev \cite{Vasilev1} and O.~I.~Zavialov \cite{Zavialov2} as well as the 1991 preprint by K.~G.~Chetyrkin \cite{Chetyrkin} elaborate in depth on combinatorial aspects of Bogoliubov's $R$-operation.  

The set of external momenta of a given graph $\Gamma$ is written $p^\Gamma$, whereas a set of independent loop momenta for $\Gamma$ is denoted by $k^\Gamma$. The latter set depends on the parametrisation of the internal edges of the graph and is a priori not uniquely defined. However, no contradiction will arise from this ambiguity in our forthcoming reasonings \cite{IZ}. To simplify notation, we omit the coupling constants and other parameters and denote the amplitude $J_\Gamma(p^\Gamma,s;g)$, corresponding to the graph $\Gamma$,  simply by $\phi(\Gamma)$.

The notion of a spinney associated to a 1PI UV divergent Feynman graph $\Gamma$ encodes the main combinatorial structure used in Bogoliubov's $R$-operation. Following \cite{CasKen}, we recall that a Feynman graph $\Gamma$ is a connected collection of lines and vertices. It is called one-particle irreducible (1PI), if it cannot be disconnected by cutting one of its internal lines. In the following we assume all Feynman graphs to be 1PI. By $\gamma \subset \Gamma$ we mean a 1PI subgraph of $\Gamma$. 

\begin{wrapfigure}{r}{0.55\textwidth}
  \begin{center}
  \vspace{-2pt}
     \IgraphBoxB\;\;  \IgraphBoxA\;\; \IgraphBox
    \vspace{-9pt}
  \end{center}
\end{wrapfigure}

A proper 1PI subgraph is any 1PI subgraph except for $\Gamma$ itself. The picture on the right, showing a $2$-point graph in $\Phi^3_6$ theory with two proper subgraphs, may help to understand the notion of disjoint subgraphs, $\gamma_i \cap \gamma_j = \emptyset$, $\gamma_i,\gamma_j \subset\Gamma$. We may draw boxes around the 1PI subgraphs $\gamma_1$ and $\gamma_2$ of $\Gamma$. Then $\gamma_1 \cap \gamma_2 = \emptyset$ means essentially that we can box both graphs at the same time in such a way that the boxes only contain 1PI subgraphs, and are neither nested, nor do they touch or overlap. This for instance implies that we can not put both subgraphs into a single box, since this box would not contain a 1PI subgraph. For a more precise description, we refer to the standard references.

We call $S$ a proper spinney of $\Gamma$, if it consists of a nonempty union of disjoint proper 1PI subgraphs, $S=\{\gamma_1, \dots, \gamma_n\}$, $\gamma_i \subset \Gamma$, $\gamma_i \cap \gamma_j = \emptyset$ for $i \neq j$. Observe that the notion of disjointness from above resolves the notorious  problem of overlapping subgraphs. The graph $\Gamma$ in the above example has the following three proper spinneys $S_{1}:=\{\gamma_1\},S_{2}:=\{\gamma_2\},S_{12}:=\{\gamma_1,\gamma_2\}$ corresponding to the three possible boxings of the subgraphs. We call the union of all such spinneys a proper wood:
$$
	W(\Gamma):=\big\{S \subsetneq \Gamma\ \big{|}
 					S=\{\mathrm{union\ of\ disjoint\ proper\ \makebox{non-overlapp.}\ 1PI\ subgraphs\ of}\ \Gamma\} \big\}.
$$
Note that by definition $\Gamma \notin W(\Gamma)$. Let us enlarge this set to $W'(\Gamma):=W(\Gamma)\cup \{\Gamma\}$ which includes the graph $\Gamma$ itself. The graph $\Gamma$ above has the following proper wood:
$$
	W(\Gamma) = \big\{\{\gamma_1\},\{\gamma_2\},\{\gamma_1,\gamma_2\} \big\}.
$$

\begin{wrapfigure}{r}{0.19\textwidth}
  \begin{center}
    \vspace{-5pt}
	\IgraphB \\[0.3cm] \IgraphC
    \vspace{-5pt}
  \end{center}
\end{wrapfigure}
 
Reducing a subgraph $\gamma$ in a graph $\Gamma$, denoted $\Gamma/\gamma$, means contracting the subgraph to a point. This extends naturally to spinneys $S \subsetneq \Gamma$, that is, for $S \in W(\Gamma)$, $\Gamma/S$ denotes the graph following from $\Gamma$ with all the 1PI UV divergent subgraphs of $S$ contracted to points. As an example we look at the above graph $\Gamma$ and reduce its spinney $S_2$. This results in the reduced $2$-loop graph $\Gamma/S_{2}=\Gamma/\gamma_2$. Reducing spinney $S_{12}$ yields the $1$-loop graph $\Gamma/S_{12}$.
 
Recall that by the external structure of a graph we mean its external lines, i.e.~indicating the external momenta (and other properties). Observe that in general the external structure of the reduced graph $\Gamma/S$ coincides with the one of $\Gamma$, so that $p^\Gamma=p^{\Gamma/S}$ (the two graphs have the same external momenta). However, it is clear that $\Gamma$ and $\Gamma/S$ are different, i.e.~the latter is so to say smaller, since it has fewer subgraphs, and therefore fewer vertices and propagators than the former. For consistency, we will assume that the momentum parametrization of the internal edges of ${\Gamma/S}$ is inherited from the parametrization of the internal edges of $\Gamma$.

We now introduce Bogoliubov's $R$- and $\bar R$-operations in the context of momentum subtraction. Applied to an amplitude $\phi(\Gamma)$ they are related by:
\begin{equation}
\label{bogoRoperation}
	R(\Gamma):= (id - M_{p^\Gamma}^{(a(\Gamma))}) \bar{R}(\Gamma).
\end{equation}
Here we use a shorthand notation for $R(\Gamma)=R(\phi(\Gamma))$ and $\bar{R}(\Gamma)=\bar{R}(\phi(\Gamma))$. In fact, $\bar{R}(\Gamma)$ is a function of the external momenta $p^\Gamma=(p_1,\ldots,p_n)$ of the graph $\Gamma$. The map $M^{(a(\Gamma))}_{p^\Gamma}$ denotes the Taylor jet of order $a(\Gamma)$ in the variables $p^\Gamma$, i.e.~in the components of the vectors $p_i$ around zero. Hence, $M_p^{(k)}$ maps a function $f(p)$, where $p:=(p_1,\ldots,p_n)$, to its Taylor expansion up to order $k$, that is, to a polynomial of maximal degree $k$ in the components of the $p_i$. Therefore $R(\Gamma)$ is the Taylor expansion of $\bar{R}(\Gamma)$ starting at order $k+1$. Observe that in (\ref{bogoRoperation}) the subtraction degree $a(\Gamma)$ depends on the graph, and will be defined further below. 

The actual combinatorial intricacies lurk behind Bogoliubov's $\bar{R}$-operation, which is defined as follows.
\begin{equation}
\label{bogoRBARoperation}
	\bar{R}(\Gamma):= \phi(\Gamma) + \sum_{S \in W(\Gamma)} \prod_{\gamma \in S} C(\gamma)\phi(\Gamma/S).
\end{equation}
The sum goes over all proper spinneys $S$ in $\Gamma$. For more details we refer the reader to \cite{Zimmermann1,Zimmermann2} as well as Lowenstein's Maryland lectures \cite{Maryland}. The $\bar{R}$-operation essentially prepares a UV divergent Feynman amplitude $\phi(\Gamma)$ in such a way that the final subtraction of the first $a(\Gamma)$ terms of its Taylor expansion yields the renormalised, i.e.~finite, amplitude. This preparation involves Bogoliubov's counterterm map $C$, applied to proper subgraphs $\gamma \subset \Gamma$. It is defined by:
\begin{equation}
\label{bogoCounterterm}
	C(\gamma):=  - M_{p^\gamma}^{(a(\gamma))}\bar{R}(\gamma).
\end{equation}
Hence, by definition, the counterterm $C(\gamma)$ is a polynomial of degree at most $a(\gamma)$ in the external momenta ${p^\gamma}$ of the 1PI subgraph $\gamma$. However, observe that some external momenta of the subgraph $\gamma \subset \Gamma$ may be linear combinations of internal and external momenta of $\Gamma$. Therefore, $C(\gamma)\phi(\Gamma/S)$ in (\ref{bogoRBARoperation}) is a shorthand notation meaning that the overall integration over internal momenta of $\Gamma$ involved in the definition of $\phi$ has to be performed properly including the polynomial $C(\gamma)$. When kept properly in mind, this notational simplification does not lead to inconsistencies, see e.g.~\cite{IZ,Smirnov}. It is clear that the counterterm $C(\gamma)$ is recursively defined, as it includes Bogoliubov's $\bar{R}$-operation applied to subgraphs $\lambda \subset \gamma$. This recursion terminates in 1PI subgraphs $\lambda$ with no proper 1PI subgraph, for which $\bar{R}(\lambda)= \phi(\lambda)$.     
The counterterm of the full graph $\Gamma$ is just $C(\Gamma):=  - M_{p^\Gamma}^{(a(\Gamma))}\bar{R}(\Gamma).$

Note that Bogoliubov's $R$-operation applied to a graph $\Gamma$ crucially depends on the Taylor jet map $M_{p^\Gamma}^{(a(\Gamma))}$, which in turn comprises the choice of the subtraction point $p^\Gamma=q$ (recall that we assumed $q=0$) as well as the subtraction degree $a(\Gamma)$. In general, the choice of the subtraction point is one of the key issues in the renormalisation process, as it is related to physically meaningful quantities. 

Recall the simple identities for the composition of Taylor jets  $M_p^{(k)}$ corresponding to a set of independent variables $p:=(p_1,\ldots,p_m)$ \cite{Smirnov}:
$$
	M_p^{(k)}\prod_{i}(id-M_{p^i}^{(k_i)}) = 0, \qquad k \leq \sum_{i}(k_i+1) - 1
$$
and
$$
	M_p^{(k)}\prod_{i}M_{p^i}^{(k_i)} = \prod_{i}M_{p^i}^{(k_i)}, \qquad k \geq \sum_{i}k_i.
$$
Here $M_{p^i}^{(k_i)}$ is the Taylor jet up to order $k_i$ corresponding to a subset $p^i \subset p$. For $p^i,p^j \subset p$ and $p^i\cap p^j = \emptyset$ one deduces from the above identities applied to the product $f(p^i)g(p^j)$:
\begin{equation}
\label{eq:key}
	M_{p^i}^{(k_i)}(f(p^i))M_{p^j}^{(k_j)}(g(p^j))
	=M_{p}^{(k_i+k_j)} \big( M_{p^i}^{(k_i)}(f(p^i))g(p^j) + f(p^i)M_{p^j}^{(k_j)}(g(p^j)) - f(p^i)g(p^j)\big).
\end{equation}

\smallskip

\begin{remark} 
This Rota--Baxter type identity was mentioned in \cite{EFGP}, where we dubbed a family of maps $M_{p}^{(k)}$ satisfying (\ref{eq:key}) a Rota--Baxter family in the context of the BPHZ method. The important point here is that $p$ is an independent set of variables and that $p^i,p^j \subset p$. For $q'$ and $q''$ two sets of dependent variables (think of momenta associated to the edges of a Feynman diagram, which may be related due to momentum conservation at vertices) in general we do not have the above Rota--Baxter type identity at hand for the product $M_{q'}^{(k')}(f(q'))M_{q''}^{(k'')}(g(q''))$.  
\end{remark}

\smallskip

For a 1PI UV divergent one loop graph $\Gamma$ we have $W'(\Gamma) =\{\Gamma\}$ as there are no proper spinneys. Therefore:
$$
	\bar{R}(\Gamma)= \phi(\Gamma),
$$
so that:
$$
	R(\Gamma)= (id - M^{(a(\Gamma))}_{p^\Gamma})\phi(\Gamma)
	\quad{\rm{and}}\quad
	C(\Gamma)=  - M_{p^\Gamma}^{(a(\Gamma))}\phi(\Gamma).
$$
Hence, at the one loop level, the renormalised amplitude results from a naive subtraction of the first $a(\Gamma)$ terms in the Taylor expansion of $\phi(\Gamma)$ in the external momenta of $\Gamma$. The natural, or minimal, choice for the number of terms to be discarded in this Taylor expansion is the overall degree of divergence of $\Gamma$, i.e.~$a(\Gamma)=\omega(\Gamma)$. As an example beyond one loop, we look again at the graph  $\Gamma= \!\!\!\begin{array}{c} \\[-0.4cm]\IgraphGamma \end{array}\vspace{-.1cm}$ in $\Phi^6_3$-theory. For convenience we index the upper 1PI $1$-loop subgraph $\gamma_2$ and the lower one by $\gamma_1$. Recall that it has the following wood $W(\Gamma)=\{ S_1:=\{\!\!\begin{array}{c} \\[-0.4cm] \IIpointiloop_1 \end{array}\!\!\},S_2:=\{\!\!\begin{array}{c} \\[-0.4cm] \IIpointiloop_2 \end{array}\!\!\},S_{12}:=\{\!\!\begin{array}{c} \\[-0.4cm] \IIpointiloop_1 \end{array}\!\!\!,\!\!\!\begin{array}{c} \\[-0.4cm] \IIpointiloop_2 \end{array}\!\!\} \}$, and therefore:
\begin{eqnarray*}
	\bar{R}(\Gamma) &=& \phi(\Gamma) 
					+ C(\!\!\begin{array}{c} \\[-0.6cm] \IIpointiloop_1 \end{array}\!\!)\phi(\Gamma/S_1)
					+ C(\!\!\begin{array}{c} \\[-0.6cm] \IIpointiloop_2 \end{array}\!\!)\phi(\Gamma/S_2)
					+C(\!\!\begin{array}{c} \\[-0.6cm] \IIpointiloop_1 \end{array}\!\!)
					    C(\!\!\begin{array}{c} \\[-0.6cm] \IIpointiloop_2 \end{array}\!\!)\phi(\Gamma/S_{12})\\
				  &=& \phi(\Gamma) 
				  	+ C(\!\!\begin{array}{c} \\[-0.6cm] \IIpointiloop_1 \end{array}\!\!)
					      \phi(\!\!\begin{array}{c} \\[-0.6cm]  \IIpointiiloop \end{array}\!\!)
					+ C(\!\!\begin{array}{c} \\[-0.6cm] \IIpointiloop_2 \end{array}\!\!)
					      \phi(\!\!\begin{array}{c} \\[-0.6cm]  \IIpointiiloop \end{array}\!\!)
					 + C(\!\!\begin{array}{c} \\[-0.6cm] \IIpointiloop_1 \end{array}\!\!)
					    C(\!\!\begin{array}{c} \\[-0.6cm] \IIpointiloop_2 \end{array}\!\!)
					      \phi(\!\!\begin{array}{c} \\[-0.6cm] \IIpointiloop \end{array}\!\!),	 
\end{eqnarray*}
where for $i=1,2$:
$$
	C(\!\!\begin{array}{c} \\[-0.6cm] \IIpointiloop_i \end{array}\!\!) = 
	- M_{p^{{\scalebox{0.5}{\IIpointiloop}}_i}}^{(a(\!\!\!\begin{array}{c} \\[-0.6cm] 
	{\scalebox{0.5}{\IIpointiloop}}_i \end{array}\!\!\!))}
	\phi(\!\!\begin{array}{c} \\[-0.6cm] \IIpointiloop_i \end{array}\!\!).
$$
This yields:
\begin{eqnarray*}
	\lefteqn{R(\!\!\!\begin{array}{c} \\[-0.4cm]\IgraphGamma \end{array}\!\!\!) = 
	(id - M_{p^{\scalebox{0.3}{\IgraphGamma}}}^{(a(\!\!\!\begin{array}{c} \\[-0.3cm] {\scalebox{0.5}{\IgraphGamma}} \end{array}\!\!\!))}) \bar{R}( \!\!\begin{array}{c} \\[-0.4cm]\IgraphGamma \end{array}\!\!) =
	 \bar{R}( \!\!\!\begin{array}{c} \\[-0.4cm]\IgraphGamma \end{array}\!\!\!)  
	 + C(\!\!\!\begin{array}{c} \\[-0.4cm]\IgraphGamma \end{array}\!\!\!)}\\
		            &=\!\!\!\!& (id - M_{p^{\scalebox{0.3}{\IgraphGamma}}}^{(a(\!\!\!\begin{array}{c} \\[-0.3cm] {\scalebox{0.5}{\IgraphGamma}} \end{array}\!\!\!))})
		            \big(   \phi(\Gamma) + C(\!\!\begin{array}{c} \\[-0.6cm] \scalebox{0.8}{\IIpointiloop}_1 \end{array}\!\!)\phi(\!\!\begin{array}{c} \\[-0.6cm]  \IIpointiiloop \end{array}\!\!)
				 + C(\!\!\begin{array}{c} \\[-0.6cm] \scalebox{0.8}{\IIpointiloop}_2 \end{array}\!\!)\phi(\!\!\begin{array}{c} \\[-0.6cm]  \IIpointiiloop \end{array}\!\!)
				 + C(\!\!\begin{array}{c} \\[-0.6cm] \scalebox{0.8}{\IIpointiloop}_1 \end{array}\!\!)
				    C(\!\!\begin{array}{c} \\[-0.6cm] \scalebox{0.8}{\IIpointiloop}_2 \end{array}\!\!)\phi(\!\!\begin{array}{c} \\[-0.6cm] \IIpointiloop \end{array}\!\!)\big).	
\end{eqnarray*}
By looking at the fine print of the above example, i.e.~by taking the momentum flow into account, we observe that the external momenta of the subgraphs are not independent a priori. Indeed, for the graph $\Gamma$ the external momenta is $p^{\scalebox{0.3}{\IgraphGamma}}=p$. Choosing a momentum flow inside $\Gamma$, and keeping momentum conservation at each vertex in mind, the two proper 1PI one loop subgraphs have the following external momenta, $p^{{\scalebox{0.5}{\IIpointiloop}_1}} = k$ for, say the lower subgraph, and $p^{{\scalebox{0.5}{\IIpointiloop}}_2} = p-k$ for the other one. Hence, the external momenta are not independent. The effects of this dependency on the algebraic structure of momentum subtraction schemes have to be considered in the light of foregoing remark, and of the identity (\ref{eq:key}).   

\medskip

{\bf{The subtraction degree $a$}}. In our Hopf algebraic approach the subtraction degree $a$ as a function on graphs, plays an important role. Following, e.g. \cite{Chetyrkin, Maryland, Smirnov, Zavialov1, Zimmermann1,Zimmermann2}, we introduce now the notion of oversubtractions. It appears in the definition of the Taylor jet $M^{(a(\Gamma))}_\Gamma\phi(\Gamma)$ and indicates the order up to which we may expand the Taylor jet. In general $a(\Gamma) \ge \omega(\Gamma)$. The following theorem for Bogoliubov's $R$-operation then holds. See e.g. \cite{Chetyrkin,Smirnov}.

\begin{thm}\label{thm:oversubtraction}
The renormalised Feynman amplitude $R(\Gamma)$ is finite if the subtraction degree $a(\gamma)$ satisfies:
\begin{equation}
 \label{oversub}
	a(\gamma) \ge  \omega(\gamma) + \sum_{\gamma_i \in S} \big(a(\gamma_i) - \omega(\gamma_i)\big)
 \end{equation}
for any proper subgraph $\gamma \in \Gamma$ and any spinney $S \in W(\gamma)$.
\end{thm}

The minimal solution to this constraint is given by defining $a(\gamma) = \omega(\gamma)$, corresponding to the case of {\it{minimal momentum subtraction}}. 

Another simple solution to (\ref{oversub}) is given by choosing for $\gamma$ any subgraph of $\Gamma$, including $\gamma=\Gamma$:
\begin{equation}
 \label{oversubTOP}
	\bar{a}(\gamma) = \omega(\gamma) + \sum_{{\gamma' \subset \gamma} \atop \gamma' \neq \gamma} \omega(\gamma'),
\end{equation}
where the sum goes over all proper connected 1PI subgraphs ${\gamma' \subset \gamma}$. Equivalently (recall that we work with power-counting renormalisable theories, i.e.~$\omega(\Gamma)=\omega(\Gamma/S)$ for any $S \in W(\Gamma)$), choosing any 1PI UV divergent proper subgraph $\gamma \subset \Gamma$:
$$
	\bar{a}(\Gamma)=\bar{a}(\Gamma/\gamma)+\bar{a}(\gamma),
$$
with $\bar{a}(\gamma)=\omega(\gamma)$ if $\gamma$ does not contain proper subgraphs, that is, if it is a 1PI UV divergent primitive graph. 
For instance, in the $\Phi_6^3$-model we find:
\allowdisplaybreaks{
\begin{eqnarray*}
 \bar{a}\big(\!\!\!\begin{array}{c} \\[-0.6cm]
                            \IIpointiloop
                             \end{array}\!\!\!\big) &=&  
                             \omega\big(\!\!\!\begin{array}{c} \\[-0.6cm]
                               \IIpointiloop
                             \end{array}\!\!\!\big)=2\\
 \bar{a}\big(\!\!\!\begin{array}{c} \\[-0.6cm]
                            \IIpointiiloop
                             \end{array}\!\!\!\big) &=&  
                             \omega\big(\!\!\!\begin{array}{c} \\[-0.6cm]
                             \IIpointiiloop
                             \end{array}\!\!\!\big)
                             + \omega\big(\!\!\!\begin{array}{c} \\[-0.6cm]
                             \IIpointiloop
                             \end{array}\!\!\!\big) = 4\\
\bar{a}\big(\!\!\!\begin{array}{c} \\[-0.5cm]
                            \IIpointiiiloop
                             \end{array}\!\!\!\big) &=&  
                             \omega\big(\!\!\!\begin{array}{c} \\[-0.5cm]
                             \IIpointiiiloop
                             \end{array}\!\!\!\big)
                             + 2\omega\big(\!\!\!\begin{array}{c} \\[-0.6cm]
                             \IIpointiloop
                             \end{array}\!\!\!\big) = 6                             
\end{eqnarray*}}
and
\allowdisplaybreaks{
\begin{eqnarray*}
 \bar{a}\Big(\!\!\!\begin{array}{c} \\[-0.4cm]
                            \IIIpointiloop
                             \end{array}\!\!\Big) &=&  
                             \omega\Big(\!\!\!\begin{array}{c} \\[-0.4cm]
                               \IIIpointiloop
                             \end{array}\!\!\Big)=0\\
 \bar{a}\Big(\!\!\!\begin{array}{c} \\[-0.4cm]
                            \IIIpointiiloopA
                             \end{array}\!\!\Big) &=&  
                             \omega\Big(\!\!\!\begin{array}{c} \\[-0.4cm]
                             \IIIpointiiloopA
                             \end{array}\!\!\Big)
                             + \omega\Big(\!\!\!\begin{array}{c} \\[-0.4cm]
                             \IIIpointiloop
                             \end{array}\!\!\Big) = 0\\
\bar{a}\Big(\!\!\!\begin{array}{c} \\[-0.4cm]
                            \IIIpointiiloopB
                             \end{array}\!\!\Big) &=&  
                             \omega\Big(\!\!\!\begin{array}{c} \\[-0.4cm]
                             \IIIpointiiloopB
                             \end{array}\!\!\Big)
                             + \omega\big(\!\!\!\begin{array}{c} \\[-0.5cm]
                             \IIpointiloop
                             \end{array}\!\!\!\big) = 2.                            
\end{eqnarray*}}
In the $\Phi_4^4$-model we find:
\allowdisplaybreaks{
\begin{eqnarray*}
 \bar{a}\big(\!\!\!\begin{array}{c} \\[-0.5cm]
                            \IVpointiloop
                             \end{array}\!\!\!\big) &=&  
                             \omega\big(\!\!\!\begin{array}{c} \\[-0.5cm]
                               \IVpointiloop
                             \end{array}\!\!\!\big)=0\\
 \bar{a}\Big(\!\!\!\begin{array}{c} \\[-0.4cm]
                            \IVpointiiloop
                             \end{array}\!\!\Big) &=&  
                             \omega\Big(\!\!\!\begin{array}{c} \\[-0.4cm]
                             \IVpointiiloop
                             \end{array}\!\!\Big)
                             + \omega\big(\!\!\!\begin{array}{c} \\[-0.5cm]
                             \IVpointiloop
                             \end{array}\!\!\!\big) = 0.
\end{eqnarray*}}
We call the case (\ref{oversubTOP}) \it critical oversubtraction\rm . This case is encountered frequently in the classical literature on the foundations of QFT, see e.g.~Zavialov \cite{Zavialov1}.


\section{Subtraction schemes}
\label{sect:schemes}

In this section we intend to propose a tentative algebraic classification of subtraction schemes. This classification is based on the exponential method of renormalization \cite{EFP}, of which we will present a modified version in the next section. 

Recall that $F$ denotes the set of 1PI UV divergent graphs in $H$. For each $\Gamma\in F$, we write $A_\Gamma$ for the (commutative, unital) target algebra of regularised Feynman rules acting on $\Gamma$. Typically, this is an algebra of functions depending on the external momenta of $\Gamma$, the parameters of the underlying QFT, e.g.~coupling constants, etc., as well as on the regularisation. See equation (\ref{eq:stone-of-contention}) for the typical structure of elements of $A_\Gamma$. Furthermore, we assume these functions to behave as polynomial or formal power series in the external momenta. 

\begin{defn} \label{def:SubtractionSchemes}
A {\rm{subtraction scheme}} $\mathcal S$ is a family of linear projectors $P_-^\Gamma$ on $A_\Gamma$, where $\Gamma$ runs over $F$. We write $P_+^\Gamma$ for the projector $id-P_-^\Gamma$. For any $\Gamma \in F$, spinney $S\in W(\Gamma)$, and $x_\Gamma \in A_\Gamma$ as well as any family $(x_\gamma)_{\gamma\in S}$ of elements of the $A_\gamma$:

\begin{enumerate}
\item[i)]
	the subtraction scheme $\mathcal S$ is called of {\rm{counterterm type}} (CT) if and only if:
\begin{equation}
\label{CT}
{\rm{CT}}:\qquad\;\;
	P_-^\Gamma\Big(\big(\prod_{\gamma\in S}P_-^\gamma(x_\gamma)\big)P_-^{\Gamma/S}(x_\Gamma)\Big)
	= \big(\prod_{\gamma\in S}P_-^\gamma(x_\gamma)\big)P_-^{\Gamma/S}(x_\Gamma)
\end{equation}

\item[ii)]
	the subtraction scheme $\mathcal S$ is called of {\rm{regular type}} (RT) if and only if:
\begin{equation}
\label{RT}
{\rm{RT}}:\qquad\;\;
	P_+^\Gamma \Big(\big(\prod_{\gamma\in S}P_+^\gamma(x_\gamma)\big)P_+^{\Gamma/S}(x_\Gamma)\Big)
	= \big(\prod_{\gamma\in S}P_+^\gamma(x_\gamma)\big)P_+^{\Gamma/S}(x_\Gamma)
\end{equation}

\item[iii)]
	and it is called of {\rm{symmetric type}} (ST) if and only if both equation {\rm{CT}}  (\ref{CT}) and equation {\rm{RT}} (\ref{RT}) hold true.
\end{enumerate}
\end{defn}

In these formulas, we used the usual shorthand notation in Bogoliubov formula : an overall integration over internal momenta of $\Gamma$ has to be performed to insure that products such as $(\big(\prod_{\gamma\in S}P_+^\gamma(x_\gamma)\big)P_+^{\Gamma/S}(x)$ do belong to $A_\Gamma$ and do not depend functionally on the internal momenta of $\Gamma$.   

Two remarks are in order. First, recall that the use of projectors reflects the idea to isolate the divergences of regularised amplitudes. Repeated application of these projectors leaves the result invariant. In fact, subtraction schemes may be characterized in terms of projectors in an algebra. Second, in the context of momentum subtraction the key point of the above classification lies in the subtraction degree. To further motivate this definition, we present three examples below. 

\begin{exam}{\rm{
In DR+MS the target space $A_\Gamma[\varepsilon^{-1},\varepsilon]]$ is an algebra of Laurent series over the perturbation parameter $\varepsilon$; the projection $P_-^\Gamma$ maps into $\varepsilon A_\Gamma[\varepsilon^{-1}]$, orthogonally to the image of $P_+^\Gamma $ in $A_\Gamma[[\varepsilon]]$. The dependence on the external momenta of a graph $\Gamma$ is encoded in the coefficients of the Laurent series and the very structure of the projection map $P_-^\Gamma$ does not depend in the end on $\Gamma$ --it is therefore usually written simply $P_-$, without reference to $\Gamma$ or to the particular algebra of regularized amplitudes on which it acts. Finally, $P_-$ and $P_+:=id-P_-$ project on two disjoint subalgebras (see e.g. \cite{Collins1}). Therefore DR+MS is a proper ST scheme.}}
\end{exam}

In fact, DR+MS is a particular case of the class of schemes where the projection map $P_-^\Gamma=P_-$ is graph independent (analytic renormalisation belongs to this class as well, for example). In that case, the scheme is CT if the image of $P_-$ is a subalgebra, it is RT if the image of $P_+=id-P_-$ is a subalgebra, and it is ST if both the images of $P_-$ and $P_+$ are subalgebras. This last condition ensures that the scheme is Rota--Baxter (RB). That is, $P_\pm$ are linear maps on the algebra $A$, and satisfy the (weight minus one) Rota--Baxter relation:
$$
	P_\pm(x)P_\pm(y) = P_\pm(xP_\pm(y)) + P_\pm(P_\pm(x)y) - P_\pm(xy)
$$
for any $x,y$ in $A$. In this case, as the RB structure can be lifted naturally to the dual of the Hopf algebra $H$, the renormalisation process amounts to a BWH decomposition in the group of characters of $H$, as was shown in \cite{CKI}. See also \cite{EFManchon}.

\ \par

Physical results must be independent of the particular choice of renormalisation scheme. Hence, final answers following from different schemes are supposed to agree upon finite renormalisation. The latter comprise finite corrections which, in most cases, are necessary to ensure that the physical quantities (scattering amplitudes, etc.) computed from the renormalised Green functions are the right ones. This is ensured by normalisation conditions involving the external momenta of the graphs, see e.g. \cite[Sect. 8.2]{IZ}. We focus now on the classical BPHZ method, involving Taylor jet subtractions, as described above.

\begin{exam}{\rm{
Let us consider the case where we have the projector $P_-^\gamma=M_{p^\gamma}^{(\omega(\gamma))}$, i.e.~Taylor jet subtraction up to the overall degree of divergence $\omega(\gamma)$ in the external momenta of the graph $\gamma$. In that case, the equation:
$$
	P_-^\Gamma\Big(\big(\prod_{\gamma\in S}P_-^\gamma(x_\gamma)\big)P_-^{\Gamma/S}(x_\Gamma)\Big)
	= \big(\prod_{\gamma\in S}P_-^\gamma(x_\gamma)\big)P_-^{\Gamma/S}(x_\Gamma)
$$
does not hold true in general. Indeed, one checks that the left hand side is at most of degree $\omega(\Gamma)$ in the external momenta of $\Gamma$, whereas the corresponding degree on the right hand side can be arbitrary, depending on the overall degrees of divergences $\omega(\gamma)$ for $\gamma \in S \in W(\Gamma)$.  

On the other hand,  the operation $P_+^\Gamma(x_\Gamma)$, where $x_\Gamma$ is considered as a function in the external momenta of $\Gamma$, subtracts the terms of the Taylor expansion of $x_\Gamma$ up to the degree $\omega(\Gamma)$. In particular, $P_+^\Gamma$ acts as the identity on $x_\Gamma$ if this Taylor expansion reduces to $0$. Hence, since $\omega(\Gamma)=\omega(\Gamma/S)$ and $M_{p^\Gamma}^{(\omega(\Gamma))}P_+^{\Gamma/S}(x_\Gamma)=0$, the Taylor expansion up to order $\omega(\Gamma)$ in the external momenta of $\Gamma$ of $P_+^{\Gamma/S}(x_\Gamma)$, is zero, it follows that:
$$
	P_+^\Gamma \Big(\big(\prod_{\gamma\in S}P_+^\gamma(x_\gamma)\big)P_+^{\Gamma/S}(x_\Gamma)\Big)
	= \big(\prod_{\gamma\in S}P_+^\gamma(x_\gamma)\big)P_+^{\Gamma/S}(x_\Gamma).
$$
In particular, the BPHZ renormalisation method with Taylor jet subtractions up to the overall degree of divergence is a RT scheme. Note that this observation is not related to identity (\ref{eq:key}). In other words, that this fundamental scheme is not ST and, in particular, not RB.}}
\end{exam}

\begin{exam} {\rm{
Let us consider now the case where the projector $P_-^\gamma$ is $M_{p^\gamma}^{({\bar a}(\gamma))}$, that is, Taylor jet subtraction in the external momenta of $\gamma$ up to the critical oversubtraction degree. In that case, the equation:
$$
	P_-^\Gamma\Big(\big(\prod_{\gamma\in S}P_-^\gamma(x_\gamma)\big)P_-^{\Gamma/S}(x_\Gamma)\Big)
	= \big(\prod_{\gamma\in S}P_-^\gamma(x_\gamma)\big)P_-^{\Gamma/S}(x_\Gamma)
$$
holds. Indeed, by definition of the critical oversubtraction degree, the two sides are of degree at most ${\bar a}(\Gamma)$ in the external momenta of $\Gamma$. On the other hand, the equation: 
$$
	P_+^\Gamma\Big(\big(\prod_{\gamma\in S}P_+^\gamma(x_\gamma)\big)P_+^{\Gamma/S}(x_\Gamma)\Big)
	= \big(\prod_{\gamma\in S}P_+^\gamma(x_\gamma)\big)P_+^{\Gamma/S}(x_\Gamma)
$$
does not hold. Therefore, the BPHZ method with Taylor jet subtractions up to the critical oversubtraction degree is a proper CT scheme. Note again that this fact is independent of particular identity (\ref{eq:key}).}} 
\end{exam}

We would like to postpone a further development of these examples. However, note that, by assuming the subtraction degree $a(\Gamma)$ to be such that for a graph $\Gamma$ and its spinneys $S \in W(\Gamma)$: 
$$
	a(\Gamma) > \sum_{\gamma\in S} a(\gamma) + a(\Gamma/S),
$$ 
we would find a CT scheme. 

Notice also that some physical subtraction schemes may not enter our classification approach. The reason for this is that, in practice, several constraints have to be taken into account properly when performing renormalisation. One may wish, for example, to use physical masses (which creates specific requirements on the subtraction points in momentum subtraction schemes), or, one may also have to take into account the existence of IR divergences. Even for such fundamental theories as quantum electrodynamics at low loop orders, these requirements create severe constraints. Therefore, at the moment, we consider this work as a necessary first step, and postpone a more elaborate study of these intricate scheme-dependent phenomena within the here promoted algebraic picture of renormalisation.


\section{Exponential Renormalisation}
\label{sect:ExpRen}

The exponential method was described in detail in \cite{EFP}. In the following we present a slightly refined picture of it, well adapted to renormalisation via momentum subtraction schemes. We assume once again that the quantum field theory under consideration is renormalisable, which implies, among other things, that for any spinney $S \in W(\Gamma)$ the overall degree of divergence of the graph $\Gamma/S$ is equal to the overall degree of divergence of $\Gamma$. 

Integrating over internal momenta of a graph defines a map $\int_\Gamma$ from tensor products $A_{\gamma_1}\otimes \ldots \otimes A_{\gamma_n}\otimes A_{\Gamma/S}$ to $A_\Gamma$, where $S=\{\gamma_1,\ldots,\gamma_n\} \in W(\Gamma)$ is an arbitrary spinney.
%
%
We deliberately avoid discussing the analytical construction of the function spaces $A_\Gamma$ and of the integration map which is not relevant for our purposes, and refer e.g.~to \cite{Zavialov2}.

Recall now that the set $F$ of 1PI UV divergent graphs of a given renormalisable quantum field theory generates a polynomial algebra $H$ that carries naturally the structure of a graded connected Hopf algebra, i.e.~$H=\bigoplus_{n \geq 0} H_n$, over the field $k$. Connectedness refers to $H_0 = k$. The unit in $H$ is denoted by $\un$. The commutative product in $H$ is simple disjoint union of graphs, denoted by concatenation. See e.g.~\cite{FGB,Manchon} for more details. The grading is given by the number of loops in graphs; we write  $\pi_{(n)}$ for the projection of $H$ to $H_n$, orthogonal to the other graded components of $H$. For $\phi$ a linear map on $H$, we write $\phi_{(n)}$ for $\phi \circ \pi_{(n)}$.
%
%
The coproduct $\Delta : H \to H \otimes H$ is given by:
$$
   \Delta(\Gamma) := \Gamma \otimes \un + \un \otimes \Gamma 
   + \sum_{S \in W(\Gamma)} \prod_{\gamma \in S} \gamma \otimes \Gamma/S,
$$
where the sum runs over all non-empty spinneys \cite{CKI}. This coproduct is coassociative $(\Delta\otimes Id_H)\circ\Delta =(Id_H\otimes \Delta)\circ \Delta$, and can be iterated to a map $\Delta^{[k]}$ from $H$ to $H^{\otimes k}$. The coassociativity property insures that this map does not depend on the order in which the coproduct operations are performed so that we can choose for example $\Delta^{[k]}:=(\Delta^{[k-1]}\otimes Id_H)\otimes \Delta$.
The action of $\Delta^{[k]}$ on $\Gamma$ can be expanded as a sum of tensor products of graphs and cographs of length at most $k+1$. 


We call from now on $F$-adapted linear form on $H$ a family of linear maps from $\CC \cdot \Gamma_1 \ldots \Gamma_n$ to $(A_{\Gamma_1}\otimes \ldots \otimes A_{\Gamma_n})_{S_n}$. Here $S_n$ is the symmetric group of $n$ elements. We write $\phi \in Lin_F(H)$ for such a family. The convolution product of two $F$-adapted linear forms $\phi$ and $\psi$ on $H$ is the $F$-adapted linear form defined by: 
$$
	(\phi \ast \psi) (\Gamma_1 \ldots \Gamma_n) 
	:=\langle \phi \otimes \psi,\Delta(\Gamma_1 \ldots \Gamma_n)\rangle 
	:= \int_{\Gamma_1} \ldots \int_{\Gamma_n}
	\phi(\Gamma_{1}^{(1)} \ldots \Gamma_{n}^{(1)})
	\psi(\Gamma_{1}^{(2)} \ldots \Gamma_{n}^{(2)}),
$$ 
where we use the Sweedler notation for the coproduct, i.e.~$\Delta(h)=h^{(1)} \otimes h^{(2)}$ and Fubini's theorem according to which the integrations over disjoint sets of internal momenta commute. Excepted for keeping track of the target spaces of Feynman rules, this construction coincides with the convolution product of \cite{CKII}, to which we refer for further details. For notational simplicity, as in the Bogoliubov formula, we do not keep track of the integrations and write simply from now on $\phi(\Gamma_1^{(1)} \ldots \Gamma_n^{(1)})\psi(\Gamma_1^{(2)} \ldots \Gamma_n^{(2)})$ for the right-hand side of the previous equation.

The convolution product endows the space $Lin_F(H)$ naturally with an unital algebra structure. A character $\phi$ on $H$ is a multiplicative unital $F$-adapted map on $H$. That is, $\phi(hg)=\phi(h)\phi(g)$, $\phi(\un)=1$. The set of characters $G$ forms a group for the convolution product. Its unit, written $e$, is the projection on $H_0=k$ orthogonally to the $H_i$, $i \geq 1$. Recall that the inverse in $G$ is given by composition with the Hopf algebra antipode.

Let us briefly give a heuristic argument for this Hopf algebraic framework in the light of Bogoliubov's classical recursion. One may use the convolution product to disentangle Bogoliubov's formulae. Indeed, if the counterterm is a character, the recursion rewrites:
 \begin{eqnarray*}
	R(\Gamma)	&=&(id-T)\big(\phi(\Gamma)+\sum_{S \in W(\Gamma)}  
	C(\prod_{\gamma \in S}\gamma)\phi(\Gamma/S)\big)\\
			&=&(id-T)(C\ast (\phi-e))(\Gamma),
\end{eqnarray*}
where we wrote $T$ for the general subtraction map. It is here that the Hopf algebraic framework may imply a restriction on the algebraic nature of the map $T$, due to the needed multiplicativity of $C$. As an example we recall that in the context of MS in DR, the Rota--Baxter relation for $T$, i.e.~$T(x)T(y)=T(xT(y)+T(x)y - xy)$, is needed to lift Bogoliubov's classical recursion to the Hopf algebraic framework \cite{CKII,Kchen}. In the following, we show that the exponential method provides a sound Hopf algebraic procedure to construct counterterms in the group of characters, which allows to incorporate momentum subtraction schemes. 

Let a subtraction scheme be associated to the target algebras $A_\Gamma$, that is assume that a subtraction map $P_-^\Gamma$ is defined on each of these algebras. The focus of the construction presented below is to deal with RT schemes (\ref{RT}) introduced in the foregoing section. Note that this is actually the reason why we depart from the recursion presented in reference \cite{EFP}, and construct recursively the counterterm character instead of the renormalised Feynman rule character.

\smallskip

The subtraction maps $P_-^\Gamma$ induce an operator ${\mathcal P}_\pm$ on $G$:
\begin{equation}
\label{KeyProperty}
	{\mathcal P}_\pm(\phi) (\Gamma_1 \cdots \Gamma_k):=
	P_\pm^{\Gamma_1}(\phi(\Gamma_1)) \cdots P_\pm^{\Gamma_k}(\phi(\Gamma_k)),
\end{equation}
and:
$$
	{\mathcal P}_\pm(\phi) (\un):=1,
$$
where $\Gamma_i \in F$, and $\phi$ is an arbitrary element in $G$. This definition is natural from the point of view of Taylor expansions of products of functions of not necessarily independent variables. 

\begin{lem}\label{rt} 
For a RT scheme, and any $\phi, \psi \in G$:
\begin{equation}
\label{RTproperty}
	{\mathcal P_+}\big({\mathcal P_+}(\phi) \ast{\mathcal P_+}(\psi)\big)={\mathcal P_+}(\phi) \ast {\mathcal P_+}(\psi).
\end{equation}
\end{lem}

\begin{proof}
Due to the multiplicativity properties of ${\mathcal P_+}$, and since ${\mathcal P_+}(\phi)$ and ${\mathcal P_+}(\psi)$ belong to $G$, it is enough to prove the identity when the operators are acting on a connected graph $\Gamma$:
\allowdisplaybreaks{
\begin{eqnarray*}
\lefteqn{
	{\mathcal P_+}\big({\mathcal P_+}(\phi)\ast{\mathcal P_+}(\psi)\big)(\Gamma) 
	= P_+^\Gamma\big({\mathcal P_+}(\phi)\ast{\mathcal P_+}(\psi)(\Gamma)\big)}\\
	&=& P_+^\Gamma\big(P_+^\Gamma(\phi(\Gamma))\big)
		+ P_+^\Gamma\big(P_+^\Gamma(\psi(\Gamma))\big)
		+ P_+^\Gamma\big( \sum_{S \in W(\Gamma)} \prod_{\gamma \in S} P_+^\gamma(\phi(\gamma)) 
		P_+^{\Gamma/S}(\psi(\Gamma/S ))\big)\\
	&\stackrel{(\ref{RT})}{=}& P_+^\Gamma(\phi(\Gamma))+P_+^\Gamma(\psi(\Gamma))
		+ \sum_{S \in W(\Gamma)} \prod_{\gamma \in S} P_+^\gamma(\phi(\gamma)) P_+^{\Gamma/S}(\psi(\Gamma/S ))\\
	&= &\big({\mathcal P_+}(\phi)\ast{\mathcal P_+}(\psi)\big)(\Gamma).
\end{eqnarray*}}
\end{proof}

\begin{defn} \label{def:regular}
A character $\varphi \in G$ is said to be regular (irregular) up to order $n$, or $n$-regular (respectively $n$-irregular), if $\mathcal{P}_+(\varphi) \circ \pi_{(l)}  = \varphi_{(l)}$ ($\mathcal{P}_-(\varphi) \circ \pi_{(l)}  = \varphi_{(l)}$) for all  $l \leq n$.  A character is called regular (irregular) if it is $n$-regular ($n$-irregular) for all $n$.
\end{defn}

\begin{defn} \label{def:BWH}
A character $\phi \in G$ admits a BWH decomposition if there exists irregular and regular characters $\phi_-$ and $\phi_+$, respectively, such that $\phi=\phi_-^{-1} \ast \phi_+$ or, equivalently $\phi_- \ast \phi = \phi_+$.
\end{defn}

In DR+MS and other RB schemes, the BWH decomposition of a character always exists and is uniquely defined. However, in general, a character may admit several BWH-like decompositions.

From now on, we focus on RT schemes, having especially in mind the example of minimal momentum subtraction ---we will explain in the conclusion of this section how the reasoning can be adapted to the other types of subtraction schemes, i.e.~CT and ST schemes.

\begin{cor}\label{cor:RTclosed}
 For a RT scheme, the convolution products of regular (resp. $n$-regular) characters are regular (resp. $n$-regular) characters.
\end{cor}

This follows from the computation in the proof of Lemma~\ref{rt}. Notice that this property does not hold for products of irregular characters.

Recall that a $F$-adapted linear map $\mu$ on $H$ is called an infinitesimal character if and only if it vanishes on $H_0$ and on all non trivial products of graphs: $\mu (\Gamma_1 \cdots \Gamma_k)=0$ whenever $k>1$, where the $\Gamma_i$s are arbitrary graphs in $F$. By general properties of graded connected commutative Hopf algebras, the convolution exponential of an infinitesimal character is a character (and conversely, the convolution logarithm of a character is an infinitesimal character, see e.g.~\cite{EFGP}). We say that an infinitesimal character $\mu$ is regular if, for an arbitrary $\Gamma\in F$, $P_+^\Gamma(\mu(\Gamma))=\mu(\Gamma)$.

\begin{lem}\label{lem:ExpLogRT}
For a RT scheme, the exponential of a regular infinitesimal character is a regular character. Conversely, the logarithm of a regular character is a regular infinitesimal character.
\end{lem}

\begin{proof}
The second statement follows from Lemma~\ref{rt}. To prove the first, since $\exp^{\ast}(\mu)$ is a character and because of the multiplicativity properties of ${\mathcal P}_+$, it is enough to prove that, for $\Gamma \in F$, ${\mathcal P}_+(\exp^{\ast}(\mu))(\Gamma)=\exp^{\ast}(\mu)(\Gamma)$. Since $\mu$ is infinitesimal, the action of $\mu^{\otimes n}$ vanishes on products $h_1\otimes ...\otimes h_n$, $h_i\in H$, whenever one of the $h_i$s is either the empty graph or a nontrivial product of 1PI graphs. We get:
$$
	{\mathcal P}_+(\exp^{\ast}(\mu))(\Gamma)
	=P_+^\Gamma(\exp^{\ast}(\mu)(\Gamma))
	=P_+^\Gamma \big(\sum_k \frac{1}{(k+1)!} \sum_{I_k} (\prod_\gamma \mu(\gamma)) \mu(\Gamma/S)\big),
$$
where in the sum, $I_k$ parametrizes all the terms $\mu(\gamma_1)\cdots \mu(\gamma_k)\mu(\gamma_{k+1})$, with $\gamma_i\in F$, showing up in the expansion of $\mu^{\ast k+1}(\Gamma)$. The notation $\gamma_{k+1}=\Gamma/S$ indicates that (by definition of the coproduct and since it is coassociative) this last graph identifies with $\Gamma$ where the elements of a given spinney have been contracted to points. The proof follows since the scheme is RT and $\mu(\gamma)=P_+^\gamma(\mu(\gamma))$ for any $\gamma$.
\end{proof}

\begin{prop}(Exponential method for RT schemes) \label{prop:Exp2}
Let $\varphi$ be a Feynman rule character, i.e.~an element of $G$. We consider the following recursion:
$$
	\varphi^-_{n+1} := \Upsilon^+_{n+1} * \varphi_{n}^- ,
$$
where $\Upsilon_{0}^+ := e$, $\varphi_{0}^-:=\varphi^{-1}$, $\Upsilon^+_{n+1}:=\exp^*(- \mu_{n+1})$ and $\mu_{n+1}:= \mathcal{P}_+(\varphi_{n}^-)\circ \pi_{(n+1)}$. Then:
\begin{enumerate}
\item
	The map $\mu_{n+1}$ is an infinitesimal character.
\item
	The map $\Upsilon^+_{n+1}$, called the order $n$ renormalisation factor, is a regular character. 
\item 
	The map $\varphi^-_{n+1}$ is a $n+1$-irregular character.
\end{enumerate}
\end{prop}

Before we prove these assertions, let us state the main theorem for the exponential method applied to RT schemes, such as e.g.~momentum subtraction up to the overall degree of divergence. Notice first that the direct limit $\varphi^- := \lim\limits_\rightarrow \varphi^-_{n}$ gives an irregular character to all orders. Moreover: 
$$
	\varphi^- = \Upsilon_{\infty}^+ * \varphi^{-1}
$$ 
or, equivalently, $\varphi^-\ast \varphi = \Upsilon_{\infty}^+$ where:
$$
	\Upsilon_{\infty}^+ := \lim\limits_\rightarrow \Upsilon^+(n),
$$
is a regular character with:
$$
	\Upsilon^+(n) := \Upsilon^+_{n}  \ast \cdots \ast \Upsilon_{1}^+ 
	                       =\exp^*(-\mu_n)*\cdots * \exp^*(-\mu_1).
$$
The character $\Upsilon^+(n)$ is called the order $n$ renormalised character.

\begin{thm}\label{thm:mainExpRenII}
With the hypothesis and the notation of the previous proposition, let us write $C$ and $R$ for Bogoliubov's counterterm character and the renormalised character, respectively. These are those characters on $H$, with values for graphs in $F$ given by Bogoliubov's recursion (\ref{bogoRoperation}),(\ref{bogoRBARoperation}). Then,
$$
	C \ast \varphi= R
$$
with $C=\varphi^-$ and $R=\Upsilon^+_\infty:= \lim\limits_{\to}\exp^*( -\mu_{n})  \ast \cdots \ast \exp^*( -\mu_{1}).$
\end{thm}

Let us prove first Proposition \ref{prop:Exp2}. By induction $\varphi^-_{n}$ is a $n$-irregular character. Therefore, for any connected graph $\Gamma$ of degree less or equal to $n$:
$$
	{\mathcal P}_-(\varphi^-_{n})(\Gamma)=P_-^\Gamma(\varphi^-_{n}(\Gamma))=\varphi^-_{n}(\Gamma).
$$
This implies that for any non empty set of connected graphs $\Gamma_1,\ldots,\Gamma_k$ of total degree $|\Gamma_1|+\cdots+|\Gamma_k|=n+1$, we get:
\allowdisplaybreaks{
\begin{eqnarray*}
	\mu_{n+1}(\Gamma_1 \cdots \Gamma_k)
	&=& \mathcal{P}_+(\varphi_{n}^- )(\Gamma_1\cdots \Gamma_k)\\
	&=&P_+^{\Gamma_1}(\varphi_{n}^- (\Gamma_1)) \cdots P_+^{\Gamma_k}(\varphi_{n}^- (\Gamma_k))\\
	&=& 0
\end{eqnarray*}}
since for any graph $\gamma$, $P_+^\gamma\circ P_-^\gamma=0$. The first assertion follows. Now, by its very definition $\mu_{n+1}$ is regular. Its convolution exponential is a regular character. The second assertion follows.

Recall that by construction $\Upsilon^+_{n+1}$ is zero on $H_i$, $0<i\leq n$. Therefore, for simple degree reasons, $\varphi^-_{n+1}= \Upsilon^+_{n+1} * \varphi_{n}^- =\varphi^-_n$ on $H_i$ for $0<i\leq n$. Indeed, in the coproduct $\Delta(h)=h^{(1)} \otimes h^{(2)}$ of any $h \in H_i$, $0<i\leq n$, we have that the degree of $h^{(1)}$ is at most $i$. Besides, $\varphi^-_{n+1}$ is a character as a product of characters. It is therefore a $n$-irregular character, and for  $\Gamma_1,\ldots,\Gamma_k$ of total degree $|\Gamma_1|+\cdots+|\Gamma_k|=n+1$ we get:
\allowdisplaybreaks{
\begin{eqnarray*}
	{\mathcal P_-}(\varphi^-_{n+1})(\Gamma_1 \cdots \Gamma_k)
	&=& {\mathcal P_-}(\varphi^-_{n})(\Gamma_1\cdots \Gamma_k)\\ 
	&=& P_-^{\Gamma_1}(\varphi^-_{n}(\Gamma_1))\cdots P_-^{\Gamma_k}(\varphi^-_{n}(\Gamma_k))\\
	&=& \varphi^-_{n}(\Gamma_1)\cdots \varphi^-_{n}(\Gamma_k)=\varphi^-_{n+1}(\Gamma_1 \cdots \Gamma_k).
\end{eqnarray*}}
Let us now compute the action of ${\mathcal P}_-(\varphi^-_{n+1})$ on a connected graph $\Gamma$ of degree $n+1$. We get, since $\Upsilon^+_{n+1}$ acts as the null map on $H_i$, $i\leq n$ and equals $-\mu_{n+1}$ in degree $n+1$:
$$
	{\mathcal P}_-(\varphi^-_{n+1})(\Gamma)=P_-^\Gamma(\varphi^-_{n+1}(\Gamma))
$$
and:
\allowdisplaybreaks{
\begin{eqnarray*}
	\varphi^-_{n+1}(\Gamma)
	&=& (-\mu_{n+1})*\varphi^-_{n}(\Gamma)\\
	&=& \varphi^-_{n}(\Gamma)-\mu_{n+1}(\Gamma)=\varphi^-_{n}(\Gamma)-\mathcal{P}_+(\varphi_{n}^- )(\Gamma)\\
	&=& \varphi^-_{n}(\Gamma)-P_+^\Gamma(\varphi_{n}^- (\Gamma))=P_-^\Gamma(\varphi_{n}^- (\Gamma)),
\end{eqnarray*}}
from which the last assertion of the proposition follows.

Let us prove now Theorem \ref{thm:mainExpRenII} recursively, and assume that, on products of graphs of degree less or equal to $n$ we have $C=\varphi_n^-$ and $R=\Upsilon^+(n)$. We get immediately (by its very definition) that, for the preparation map acting on an $n+1$ loop graph $\Gamma$ (i.e. degree $n+1$ graph):
$$
	{\bar R}(\Gamma)=(\varphi_n^-\ast (\varphi- e))( \Gamma).
$$
Note that the product $\varphi_n^-\ast (\varphi- e)$ is not a character. Since $\varphi_n^-=\Upsilon^+(n)\ast\varphi^{-1}$, we get finally:
$$
	{\bar R}(\Gamma)=\Upsilon^+(n)(\Gamma)-\varphi_n^-(\Gamma).
$$
According to the Bogoliubov recursion: 
$$
	C(\Gamma)=-P_-^\Gamma({\bar R}(\Gamma))
			  =-P_-^\Gamma(\Upsilon^+(n)(\Gamma)) + P_-^\Gamma(\varphi_n^-(\Gamma)).
$$
However,  since $\Upsilon^+(n)$ is regular, $\Upsilon^+(n)(\Gamma)=P_+^\Gamma(\Upsilon^+(n)(\Gamma))$ and the first term vanishes. We get:
$$
	C(\Gamma)=P_-^\Gamma(\varphi_n^-(\Gamma))
			  =\varphi_n^-(\Gamma)-P_+^\Gamma(\varphi_n^-(\Gamma))
			  =\varphi_{n+1}^-(\Gamma).
$$
Similarly:
\allowdisplaybreaks{
\begin{eqnarray*}
	R(\Gamma)
	&=& P_+^\Gamma(\bar R(\Gamma))=\Upsilon^+(n)(\Gamma)-P_+^\Gamma(\varphi_n^-(\Gamma))\\
	&=&\Upsilon^+(n)(\Gamma)-\mu_{n+1}(\Gamma)\\
	&=&\Upsilon^+(n)*\Upsilon_{n+1}^+(\Gamma)=\Upsilon^+(n+1)(\Gamma),
\end{eqnarray*}
and the proof of the theorem is complete.

The above construction can be dualised for CT schemes. By a symmetry argument, the recursion focusses in that case on the counterterm instead of the renormalised character ---this is the recursion dealt with in \cite{EFP}. 

For ST schemes, which generalises the notion of RB schemes, the picture is even simpler. Indeed, for such schemes the BWH decomposition is necessarily unique. Assume that two such decompositions of the character $\phi$ exist:
$$
	\phi_-\ast\phi_+=\psi_-*\psi_+,
$$
then, by Lemma~\ref{rt} (and the analogous Lemma for CT schemes), products of characters preserve regularity and irregularity properties; $\psi_-^{-1}$ is irregular and $(\phi_+)^{-1}$ is regular. Therefore:
$$
	(\psi_-)^{-1}\ast\phi_-=\psi_+*(\phi_+)^{-1}
$$
is an identity between an irregular and a regular character. It follows that $(\psi_-)^{-1}\ast\phi_-=e=\psi_+*(\phi_+)^{-1}$, and the unicity property follows. Any recursive construction of a counterterm and of a renormalised character leads therefore in that case to the same unique solution, that agrees necessarily with the solution provided by the $R$ operation.

\medskip

In the light of these algebraic developments, we would like to add a remark. One of the advantages of a detailed algebraic description of the process of perturbative renormalisation is an improved understanding of possible freedoms, respectively variations, in the construction of its main ingredients, i.e.~the counterterm and renormalised characters. Moreover, this becomes relevant in the context of applications of renormalisation techniques beyond the usual QFT domain, especially in aspects related to mathematical questions motivated by physics (e.g.~singularities of hypergeometric functions, rough paths, etc.).

Focussing for example once again on RT schemes, we consider now another way of constructing a BWH type decomposition that differs in general from the one we considered previously (and therefore also from the solution to the Bogoliubov recursion). It would be interesting to investigate the properties of this recursion on concrete examples, a task that we postpone to further research, dealing in the present article with the formal algebraic aspects of the theory.

The second exponential method for RT schemes, we would like to propose works very similar to the one in Proposition \ref{prop:Exp2}. Indeed, let $\varphi \in G$. We consider the following recursion $ \varphi^-_{n+1}:= \varphi_{n}^- *\Upsilon^+_{n+1}$, where $\Upsilon_{0}^+ := e$, $\varphi_{0}^-:=\varphi$, $\Upsilon^+_{n+1}:=\exp^*(- \mu_{n+1})$ and $\mu_{n+1}:= \mathcal{P}_+(\varphi_{n}^- )\circ \pi_{(n+1)}$. Then, the map $\mu_{n+1}$ is a regular infinitesimal character and the renormalisation factors $\Upsilon_{n+1}^+ :=\exp^*(- \mu_{n+1})$ are regular characters. The map $\varphi^-_{n+1}$ is a $n+1$-irregular character. Notice that the direct limit $\varphi^- := \lim\limits_\rightarrow\varphi^-_{n}$ gives an irregular character to all orders. Hence, we have the BWH decomposition, $\varphi =\varphi^-\ast (\Upsilon_{\infty}^+ )^{-1}$, where $\Upsilon_{\infty}^+:=\lim\limits_\rightarrow \Upsilon^+(n),$ is a regular character with $\Upsilon^+(n):=  \Upsilon_{1}^+ \ast \cdots \ast\Upsilon^+_{n}$. The inverse character $(\Upsilon^+(n))^{-1}=\exp^*(\mu_1)*\cdots *\exp^*(\mu_n)$ is called the $n$ order renormalised character.

This existence of several different recursions to construct (ir-)regular characters suggests interesting algebraic structures related to the perturbative process of renormalisation. A coherent description of the physical (and possibly mathematical) relevances is a challenging task, and we postpone further exploration to future work.


\subsection*{Acknowledgments}

We would like to thank J.~M. Gracia-Bond\'{i}a for helpful discussions. The first author was supported by a Juan de la Cierva postdoctoral research grant from the Spanish Government. Both authors would like to thank the CNRS (GDR Renormalisation) for support.


\end{document}